\documentclass[conference]{IEEEtran}
\IEEEoverridecommandlockouts

\usepackage[tbtags]{amsmath} % defines many math commands and
\usepackage{amssymb}  % get, among others, blackboard bold fonts
\usepackage{verbatim} % get comment environment, + new verbatim

\usepackage{enumerate}

\usepackage{amsfonts}
\usepackage{color}

\usepackage{amsthm}

\usepackage{cite}

\usepackage{tikz}

\usetikzlibrary{shapes,arrows,decorations.markings,positioning}
\usetikzlibrary{fit,backgrounds}

\tikzstyle{plant} = [draw, fill=red!5, rounded rectangle, 
    minimum height=3em]
\tikzstyle{block} = [draw, fill=blue!5, rectangle, 
    minimum height=3em]
\tikzstyle{tap} = [draw, fill=red!5, rectangle, minimum height=3em]
\tikzstyle{sum} = [draw, fill=yellow!10, circle, node distance=1cm]
\tikzstyle{pinstyle} = [pin edge={to-,thick,black}]

\tikzstyle{BitPipe} = [thick, decoration={markings,mark=at position
   1 with {\arrow[semithick]{open triangle 60}}},
   double distance=1.4pt, shorten >= 5.5pt,
   preaction = {decorate},
   postaction = {draw,line width=1.4pt, white,shorten >= 4.5pt}]

\usetikzlibrary{chains,shapes.multipart}
\tikzstyle{FIFO} = [rectangle split, rectangle split parts=3, draw, rectangle split horizontal,minimum height=3em,text height=1.5em,text depth=1em,on chain,inner ysep=0pt]

\usepackage{epsfig, graphicx, psfrag}

\usepackage{soul}

\usepackage[cal=boondox]{mathalfa}

\newtheorem{lemma}{Lemma}

\theoremstyle{definition}

\newtheorem*{defn*}{Definition}

\newtheorem*{scheme*}{Scheme}

\theoremstyle{remark}
\newtheorem{remark}{Remark}

\providecommand{\secref}[1]{Sec.~\ref{#1}}

\providecommand{\figref}[1]{Fig.~\ref{#1}}

\newcommand{\ie}{i.e.}
\newcommand{\eg}{e.g.}

\newcommand{\etal}{et al.}
\newcommand{\cf}{cf.}

\newcommand{\nats}{\mathbb{N}}

\newcommand{\e}{\mathrm{e}}

\newcommand{\cX}{\mathcal{X}}

\newcommand{\cI}{\mathcal{I}}

\newcommand{\cE}{\mathcal{E}}

\newcommand{\Comment}[1]{}
\newcommand{\old}[1]{}
\newcommand{\rem}[1]{}

\newcommand{\oJ}{\bar{J}}

\providecommand{\e}{{\rm e}}
\providecommand{\comment}[1]{}

\newcommand{\beqn}[1]{\begin{eqnarray}\label{#1}}
\newcommand{\eeqn}{\end{eqnarray}}
\newcommand{\beq}[1]{\begin{equation}\label{#1}}
\newcommand{\eeq}{\end{equation}}

\providecommand{\var}[1]{{\mathrm{Var}\left( #1 \right)}}

\makeatletter
\newcommand{\vast}{\bBigg@{4}}
\newcommand{\Vast}{\bBigg@{5}}
\makeatother

\providecommand{\E}[1]{\mathbb{E} \left[ #1 \right]}

\usepackage{ifthen}
\usepackage{mathtools}
\mathtoolsset{showonlyrefs=true} 

\providecommand{\first}[2]{#1}

\first{\usepackage{hyperref}}{\usepackage{url}}

\first{
\usepackage{tikz}
\usetikzlibrary{shapes,arrows,decorations.markings,positioning}
}{}

\usepackage{varwidth}

\def\rchannel{r_\mathrm{ch}}
\def\rcode{r_\mathrm{code}}
\begin{document}
%%%%%%%%%%%%%%%%%%%%%%%%%%%%%%%%%%%%%%%%%%%%%%%%%%%%%%%%%%
%%%%%%%%%%%%%%%%%%%%%%%%%%%%%%%%%%%%%%%%%%%%%%%%%%%%%%%%%%
%%%%%%%%%%%%%%%%%%%%%%%%%%%%%%%%%%%%%%%%%%%%%%%%%%%%%%%%%%

\title{Real-Time Variable-to-Fixed Lossless Source Coding of Randomly Arriving Symbols}

\author{Uri Abend and Anatoly Khina
	\thanks{This research was supported by the \textsc{Israel Science Foundation} (grant No.\ 2077/20).
	The work of U.~Abend was supported by the Yitzhak and Chaya Weinstein Research Institute for Signal Processing and by the Magbit Foundation of Greater Los Angeles.
    The work of A.~Khina was supported by the 5GWIN Consortium through the Israel Ministry of Economy and Industry.}
    \thanks{The authors are with the Department of Electrical Engineering--Systems, Tel Aviv University, Tel Aviv 6997801, Israel (e-mails: \texttt{uribinyamina@mail.tau.ac.il,anatolyk@eng.tau.ac.il}).}
}

\maketitle
% %% To avoid page numbering
% \thispagestyle{plain}
% \pagestyle{plain}
%%%%%%%%%%%%%%%%%%%%%%%%%%%%%%%%%%%%%%%%%%%%%%%%%%%%%%%%%%%%%%%%%%%%%%%%%%%%%%%%%%%%%%%%%%%%%%%%%%%%%%%%%%
% \ver{\vspace{-\baselineskip}}{}
\begin{abstract}
We address the recently suggested problem of causal lossless coding of a randomly arriving source samples. 
We construct variable-to-fixed coding schemes and show that they outperform the previously considered fixed-to-variable schemes when traffic is high both in terms of delay and Age of Information by appealing to tools from queueing theory. 
We supplement our theoretical bounds with numerical simulations.
\end{abstract}

\first{
}{}%{\vspace{-\baselineskip}}

\allowdisplaybreaks

%%%%%%%%%%%%%%%%%%%%%%%%%%%%%%%%%%%%%%%%%%%%%%%%%%%%%%%%%%%%%%%%%%%%%%%%%%%%%%%%%%%%%%%%%%%%%%%%%%%%%%%%%%

\section{Introduction}
\label{s:intro}

Real-time communications has become a topic of growing interest in recent years due to rising demand in applications such as vehicular communications \cite{papadimitratos2009vehicular}, telemedicine \cite{bodner2004early} and satellite control \cite{lovera2002periodic}. Under this regime, the objective is to convey messages in a timely fashion, in contrast to classic communications where relatively long delays may be tolerated.

Two figures of merit have been suggested to measure the performance of real-time communication systems: \textit{Age of Information} (AoI) \cite{mayekar2020optimal,kaul2012real,kosta2017age} and \textit{delay} \cite{larmore1989minimum,mayekar2020optimal}.
The AoI of a system quantifies the freshness of the data at the receiver by assigning an ``age" to the most recent update of the system and measuring the time difference between updates. Delay, on the other hand, quantifies the timeliness of a system by measuring the elapsed end-to-end time difference from a symbol arrival to the encoder to its decoding at the decoder. These two quantities are similar in nature and indeed in some schemes \cite{zhong2016timeliness,ZhongYatesSoljanin:Huffman_for_AoI:ITW2018}, they can be simultaneously minimized.

In this work, we concentrate on the recently suggested setting of lossless transmission of causal source samples with random arrival times \cite{ZhongYatesSoljanin:Huffman_for_AoI:ITW2018,CausalPM:ISIT2019}. This setting is simple to formulate and can apply to many different settings, e.g., a multi-user scenario where data is transmitted through a single channel.

So far, the majority of works dealing with real-time lossless source coding, have focused on one-to-variable encoding, following the work of Larmore \cite{larmore1989minimum}. Additionally, for deterministic (periodic) arrivals, some work has been done on the more general case of block-to-variable coding \cite{zhong2016timeliness}. These coding schemes, while effective in low-traffic scenarios, exhibit large delays as the traffic increases, and the waiting time becomes dominant. However, in the case of deterministic arrivals, a variable-to-variable scheme in which the block size is set on the fly according to the number of symbols waiting to be coded, can be used to facilitate lower latency as was suggested in \cite{zhong2017backlog}.

To alleviate this problem, we propose in this work a variable-to-fixed coding scheme.
The motivation for variable-to-fixed codes stems from two main reasons: First, a work in queueing theory by Rogozin \cite{rogozin1966some} states that, under some restrictions, the waiting time of a queue is minimized by a deterministic service time. While this does not prove that variable-to-fixed codes are optimal, it suggests that they might perform well when the waiting time is the dominant component of the AoI and the delay, \ie, in high-traffic scenarios. Secondly, fixed-length codes synergize well with error correcting codes when real-time transmission is carried over noisy channels.

The rest of the paper is organized as follows. \secref{s:model} presents the system model and objectives. \secref{s:background} provides necessary background from queueing theory. In \secref{s:AoI-V2F}, we provide an analysis of the delay for variable-to-fixed codes, with \secref{s:waiting time} dedicated to an analysis of the waiting time. \secref{s:numeric} contains some simulation results. We conclude with a summary of the work and suggestions for future work in \secref{s:summary}.

%%%%%%%%%%%%%%%%%%%%%%%%%%%%%%%%%%%%%%%%%%%%%%%%%%%%%%%%%%%%%%%%%%%%%%%%%%%%%%%%%%%%%%%%%%%%%%%%%%%%%%%%%%

\section{Problem Setup}
\label{s:model}

We formulate here the real-time source coding setup that will be treated in this work, depicted in \figref{fig: block diagram}. Operation of the system is set to begin at time 0.

\begin{figure}[t]
    \resizebox{\columnwidth}{!}{\begin{tikzpicture}[auto, arrow/.style={very thick, ->, >=stealth'},start chain=going right,>=latex,node distance=.13\columnwidth,>=latex']
    \node[coordinate] (input) {};
    \node[block, right of = input, node distance = .3 \columnwidth] (enc) {Encoder};
    %\node[plant, right = of enc] (fifo) {FIFO};
    % \draw (0,0) -- ++(3em,0) -- ++(0,-3em) -- ++(-3em,0);
    % \foreach \i in {1,...,4}
    %   \draw (3em-\i*10pt,0) -- +(0,-3em);
    \node[FIFO, right=of enc] (fifo) {FIFO};
    % Erase left line of FIFO
    \draw[yellow!10,very thick] (fifo.south west) -- (fifo.north west);
    % the rectangle with vertical rules
     %\node[draw,right=of enc, rectangle,on chain,minimum size=1.5cm] (rr) {$I$};
     %\node[draw,rectangle,on chain,draw=white,minimum size=1.3cm]{};
     % the rectangular shape with vertical lines
    %  \node[rectangle split, right=of enc, rectangle split parts=5,
    %  draw, rectangle split horizontal,text height=1cm,text depth=0.5cm,on chain,inner ysep=0pt] (fifo) {FIFO};
    %  \fill[white] ([xshift=-\pgflinewidth,yshift=-\pgflinewidth]fifo.north west) rectangle ([xshift=-15pt,yshift=\pgflinewidth]fifo.south);
     % the circle
     \node[draw,circle,on chain,minimum size=1.3cm, right=0cm of fifo] (server) {Server};
     % the arrows and labels
     %\draw[->] (se.east) -- +(20pt,0) node[right] {$\mu$};
     %\draw[<-] (wa.west) -- +(-20pt,0) node[left] {$\lambda$};
     %\node[align=center,below] at (rr.south) {Input \\ process};
     %\node[align=center,below] at (wa.south) {Queue \\ subsystem};
     %\node[align=center,below] at (se.south) {Server \\ process};
     \node[block, right=of server] (dec) {Decoder};
    \node[coordinate, node distance = .2\columnwidth, right=of dec] (output) {};
    
    \draw[arrow] (input) -- node [tap, pos=0.1, scale=0.5, minimum width= .4cm, above] {}  node [pos=0.1, below] {\footnotesize{ $A_3$}} node[pos=.1, above, minimum height = 14mm] {\footnotesize $X_3$}  
    node [tap, pos=0.5, scale = 0.6, minimum width = 0.32cm, above] {} node [pos=0.5, below] {\footnotesize{ $A_2$}} node[pos=.5, above, minimum height = 16.5mm] {\footnotesize $X_2$}
    node [tap, pos=0.72, scale=0.4, minimum width = .485cm, above] {} node [pos=0.72, below] {\footnotesize{ $A_1$}} node[pos=.72, above, minimum height = 12mm] {\footnotesize $X_1$}
    % node[pos=.95, below] {\footnotesize $t$}
    node[above, pos=.7]{\begin{tabular}{l}
          Randomly arriving
       \\ source symbols
          \\ \\ \\
     \end{tabular}}
    (enc);
    %node [above] {\begin{tabular}{l} Causal \\  Source \end{tabular}} (enc);
    \draw[BitPipe] (enc) -- node [below] {$\rcode$} node [above] {..001..} (fifo);
    %\draw[arrow] (fifo) -- node {$r$ bits} (queue);
    \draw[BitPipe] (server) -- node [below] {$\rchannel$} node [above] {..010..} (dec);
    \draw[arrow] (dec) -- 
    %node [below] {$X_n$}
    node [tap, pos=0.2, scale=0.6, minimum width= .32cm, above] {}  node [pos=0.2, below] {\footnotesize{ $R_2$}}
    node[pos=.2, above, minimum height = 16.5mm]{\footnotesize $X_2$} 
    node [tap, pos=0.75, scale=0.4, minimum width = .485cm, above] {} node [pos=0.75, below] {\footnotesize{ $R_1$}} node[pos=.75, above, minimum height = 12mm] {\footnotesize $X_1$}
    node [above, pos=0.3] {\begin{tabular}{l} Causally-decoded \\ source symbols \\ \\ \\ \end{tabular}} (output);

    \begin{pgfonlayer}{background}
        \node[fill = yellow!10, draw, dashed, rounded corners, fit = (fifo) (server)] (channel) {};
    \end{pgfonlayer}
    \node[above of = channel]{Channel};
    %  \node[above of =channel]{\begin{tabular}{l}
    %       Real-time \\
    %       transmission \\ \\
    %  \end{tabular}};
\end{tikzpicture}}
    \caption{Block diagram of the system model.} 
    \label{fig: block diagram}
\end{figure}
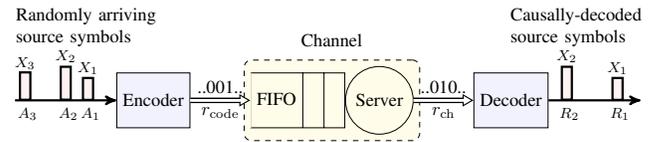

\textit{Source}:
At each (positive) time step, the source generates a new sample with known probability $q \in (0,1)$. 
The generation process samples across time are assumed independent and identically distributed (i.i.d.). 
We define the $n$-th arrival time $A_n$ as the generation time of the $n$-th symbol.
We further define the $n$-th time difference between arrivals by 
\begin{align}
    D_n = A_n - A_{n-1},
\end{align}
with $A_0$ set to $0$.
Clearly, $\{D_n\}_{n=1}^\infty$ are i.i.d.\ according to a geometric distribution with parameter $q$.

The sample values, $\{X_n\}_{n=1}^\infty$, are i.i.d.\ and are drawn from a finite alphabet $\cX$ according to a known probability mass function (p.m.f.) $p_X$. 

\textit{Encoder}: 
The memoryless encoder accumulates $B_i \in \nats$ samples and maps them into $L_i \in \nats$ bits using a prefix-free one-to-one mapping ($i$ being the block index):
\begin{align}
    \cE: \cI \subseteq \cX^* &\to \{0,1\}^*, 
\end{align}
such that $\cX^* \subseteq \bigtimes_{i=1}^\infty \cI$, \ie, any stream of source symbols is mapped to a stream of code bits.
The encoder then sends the $L_i$-length codeword over the channel. Note that both $B_i$ and $L_i$ are finite with probability 1 and that  $\{B_i\}_{i=1}^\infty$ and $\{L_i\}_{i=1}^\infty$ have i.i.d.\ samples.
The average rate of the memoryless encoder (average number of coded bits per symbol) is given~by 
\begin{align}
    \rcode = \frac{\E{L}}{\E{B}},
\label{eq:code-rate}
\end{align}
where $B$ and $L$ are general samples of their respective processes.

It is customary to divide the class of lossless encoders into three families:
\begin{itemize}
\item  
    \textit{Variable-to-fixed.} This is the encoding scheme used in this work. The output of this encoder is of fixed length $L_i \equiv \ell$. 
    The best known representative of this family is the Tunstall code \cite{TunstallPhD}, featuring a memoryless encoder which minimizes $\E{B}$ given $\ell$;
    this code will be discussed and used in later sections.
    
\item 
    \textit{Fixed-to-variable.} This family of codes was considered previously for the problem at hand in \cite{zhong2016timeliness,ZhongYatesSoljanin:Huffman_for_AoI:ITW2018,zhong2017backlog}. The input length of this encoder is of fixed length $B_i \equiv b$. 
    The best known representative of this family is the Huffman code \cite[Ch.~5.6]{CoverBook2Edition}, featuring a memoryless encoder which minimizes $\E{L}$ given $b$.

\item 
    \textit{Variable-to-variable.} This is the most general family in which both the input and the output lengths are not fixed and depend on the encoded sequence.
    The most famous code in this family is the arithmetic code \cite[Ch~13.3]{CoverBook2Edition}. Note that an encoder employing this code does not fall under the definitions given here.
\end{itemize}

\textit{Channel}: 
The channel is modeled by a first-in, first-out (FIFO) queue, which receives codewords from the encoder and outputs them to the decoder at a fixed known rate $\rchannel \in \mathbb{R}$. From a queueing perspective, we can view the codewords as clients entering the queue. This will be analyzed 
in \secref{s:AoI-V2F}.
\begin{remark}
    In general, the rate of emission from the FIFO may be any positive real number. This is the case for asynchronous FIFOs which have different input and output clock domains.
\end{remark}

\textit{Decoder}: Once $L_i$ bits are received from the FIFO, they are mapped back to $B_i$ symbols, applying the inverse of $\cE$
\begin{align}
    \cE^{-1}: \{0 ,1\}^* \to \cX^*.
\end{align}
Because the channel is noiseless and the encoder is prefix-free and lossless, we are guaranteed perfect reconstruction.
The time of decoding of the $n$-th symbol is denoted $R_n$. Note that the decoding time is shared between $B_i$ symbols. We shall denote the shared decoding time by $R_i$.
\begin{remark}
    Throughout this work we assume that the decoder is privy to the FIFO's status, e.g., via a side-channel or a special symbol reserved for this scenario, as is common in FIFO architecture.
\end{remark}

\textit{Delay}: The delay of the $n$-th symbol, $\Delta_n =  R_n-A_n$, is defined as the difference between its decoding time and arrival time.
The mean delay of the system can then be defined by averaging over all symbol delays
\begin{align}
    \bar{\Delta} = \lim_{n \to \infty}\frac{1}{n}\sum_{k=1}^n \Delta_k.
\label{eq:mean-delay}
\end{align}

\textit{AoI:}
The (peak) AoI of block $i$, $\Gamma_i = R_i - A_{i-1}$, is defined as the time difference between said decoding time of block $i$ and the arrival of the most recent symbol in block $i-1$. \footnote{Some works, e.g.,  \cite{zhong2016timeliness,zhong2017backlog}, study the average AoI. Both peak and average AoI exhibit similar behavior and we therefore treat only the former in the interest of space.}
The mean peak AoI is then given by
\begin{align}
    \bar{\Gamma} = \lim_{i \to \infty}\frac{1}{i}\sum_{k=1}^i \Gamma_k.
\label{eq:mean-aoi}
\end{align}
We refer the reader to \cite{ZhongYatesSoljanin:Huffman_for_AoI:ITW2018} for a graphical illustration of AoI. %this quantity.

\textit{Objectives:}
The objective of this encoding scheme is to minimize either the delay or the AoI. %of these two quantities. 
Note that for the case of fixed-to-variable encoding, the two quantities are simultaneously minimized.

% %%%%%%%%%%%%%%%%%%%%%%%%%%%%%%%%%%%%%%%%%%%%%%%%%%%%%%%%%%%%%%%%%%%%%%%%%%%%%%%%%%%%%%%%%%%%%%%%%%%%%%%%%%

\section{Background: Results from Queueing Theory}
\label{s:background}

We now introduce some notations and known results from queueing theory to be used later in our analysis. We define the inter-arrival time $T_i$ as the time passing between the arrival of client $(i-1)$ and client $i$, and the service time $S_i$---as the amount of time client $i$ spends inside the queue. 

\begin{remark}
    The definitions given here hold for any distributions of $T_i$ and $S_i$ with finite means \cite{lindley1952theory}. 
    For example, $T_i$ may receive only integer values and $S_i$ may receive values that are multiples of $1/\rchannel$ as described in \secref{s:model}.
\end{remark}

We are interested in results for the GI/GI/1 queue \cite{kingman1962some}, \ie, the setting of a single server, and independent i.i.d.\ processes $\{T_i\}$ and $\{S_i\}$ with known (possibly different) distributions. General samples of theses processes are denoted by $T$ and $S$.

We now present some known results for the GI/GI/1 queue waiting time, $W_i$, which is defined as the amount of time that client $i$ waits in order to enter the queue. 

To that end, first define $U_i$ as the difference between the service and arrival times of client $i$
\begin{align}
\label{eq:def:U}
    U_i \triangleq S_i-T_i.
\end{align}
Note that $U_i$ is also an i.i.d.\ process; we denote a general sample from the process as $U$.
The waiting time of the $i$-th client, $W_i$, is given by
\begin{align}
    W_i = \max(0, W_{i-1} + U_{i-1}),
\label{eq:def:waiting_time}
\end{align}
with $W_1$ set to $0$.
A queue is said to be \textit{stable} if the waiting time series $\{W_i\}_{i=1}^\infty$ is bounded with probability 1.
It was shown in \cite{kingman1962some} that a queue is stable iff\footnote{\cite{kingman1962some} showed that the queue is stable if \eqref{eq:queue_stability} holds. The other direction is trivial by noting that $W_i \geq \sum_{k=0}^{i-1} U_k$ by \eqref{eq:def:waiting_time}.}
\begin{align}
    \E{U} < 0.
    \label{eq:queue_stability}
\end{align}
Moreover,
if the queue is stable, then the waiting time series tends to a random variable, $W$, that is bounded with probability 1. Finding the mean value of $W$ requires solving integral equations that generally have no analytic solutions \cite{kingman1962some}. To overcome this, several upper bounds for $\E{W}$ have been derived in the literature.
We now present two of them, which will be used throughout the paper. The first bound, which will be refererd to as the \textit{low-moment bound}, is given by \cite{kingman1962some}
\begin{align}
    \E{W} \leq \frac{\var{U}}{-\E{U}}.
    \label{eq:kingman-bound:basic}
\end{align}
To derive the second bound, we start by defining the moment generating function (m.g.f.) of $U$:
\begin{align}
    \phi_U(\theta) = \E{e^{\theta U}}.
\label{eq:mgf:U:def}
%  \\[-1.5\baselineskip]
\end{align}
Then, the \textit{m.g.f.\ bound} is given by 
\begin{align}
    \E{W} \leq \frac{1}{\nu} \:, 
    \label{eq:kingman-bound:nu}
\end{align}
where $\nu \triangleq \sup \left\{ \theta>0 \middle| \phi_U(\theta)<1 \right\}$.

%%%%%%%%%%%%%%%%%%%%%%%%%%%%%%%%%%%%%%%%%%%%%%%%%%%%%%%%%%%%%%%%%%%%%%%%%%%%%%%%%%%%%%%%

\section{Delay and AoI of Variable-to-Fixed Codes}
\label{s:AoI-V2F}

As was stated in \secref{s:model}, a memoryless variable-to-fixed binary encoder
is a function that encodes a block of symbols of varying length $B$ which depends on the symbol values, 
to a binary codeword of fixed length $\ell$.
The coding rate \eqref{eq:code-rate}
specializes, therefore, to
\begin{align}
\label{eq:v2f:code-rate}
    \rcode = \frac{\ell}{\E{B}}.
\end{align}

From a queuing perspective, we can think of the noiseless channel as a queue with the encoder inserting clients with a fixed service time 
\begin{align}
    S_i \equiv \frac{\ell}{\rchannel},
\label{eq:v2f_service_time}
\end{align} 
and an inter-arrival time of
\begin{align}
    T_i = \sum_{n=1}^{B_i} D_n.
    \label{eq:v2f_inter_arrival_time}
\end{align}
Note that $\{D_n\}_{n=1}^\infty, \{B_i\}_{i=1}^\infty$ are independent processes. This will be used later in our analysis.

By substituting \eqref{eq:def:U}, \eqref{eq:v2f:code-rate}, \eqref{eq:v2f_service_time}, and \eqref{eq:v2f_inter_arrival_time} in \eqref{eq:queue_stability} we arrive at the following stability condition: The queue is stable iff
\begin{align}
    \rcode < \frac{\rchannel}{q},
    \label{eq:system_stability}
\end{align}
where $\E{T} = \E{B}/q$ by Wald's identity \cite{wald1944cumulative}.

\begin{remark}
    This result is similar to the one-to-variable scheme ($b = 1$ and random $L_i$)
    of
    \cite{ZhongYatesSoljanin:Huffman_for_AoI:ITW2018}, where the code rate \eqref{eq:code-rate}, 
    which specializes to $\rcode = \E{L}$ for that scheme, 
    was to be kept smaller than $\rchannel/q$ to maintain stability.
\end{remark}

Furthermore, because the encoder is lossless, the code rate satisfies $H(X) \leq \rcode$, where $H(X)$ denotes the source entropy \cite[Ch.~5.3]{CoverBook2Edition}.
Consequently, $\rchannel/q$ of a stabilizable system must be bounded from below by the source entropy 
\begin{align}
    H(X) < \frac{\rchannel}{q},
    \label{eq:fundamental_stability}
\end{align}
which leads to the following lemma.

\begin{lemma}
    The AoI and delay are bounded iff \eqref{eq:fundamental_stability} holds.
\end{lemma}

\begin{proof}
    The necessity of \eqref{eq:fundamental_stability} follows from \eqref{eq:system_stability} and the aforementioned bound $H(X) \leq \rcode$. To prove the sufficiency, consider a sequence of 
    Tunstall codes with increasing blocklength $\ell$. The rate of this sequence is known to converge to the entropy of the source \cite{TunstallPhD}:
    \begin{align}
        \lim_{\ell \to \infty} \rcode = H(X).
    \end{align}
    Therefore, \eqref{eq:system_stability} is satisfied for a Tunstall code with a sufficiently large $\ell \in \nats$. 
\end{proof}

We now return to the system objectives described in \secref{s:model}. 
To that end, we describe the delay and the AoI, each, as the sum of three known quantities, which are defined next.
    
\textit{Tarry time:} The tarry time of the $\ell$-th symbol, $J_\ell$, is defined as the amount of time it waits to be coded.
To inject bits into the channel, the encoder has to accumulate $B$ symbols. Consequently, $J_B = 0$, whereas for other symbols:
\begin{align}
    J_{\ell} &= \sum_{n=\ell+1}^B D_n, & \ell \in \{1, \ldots, B-1\}.
\label{eq:tarry:symbol}
\end{align}
The mean tarry time, $\oJ$, is equal to
\begin{subequations}
\label{eq:tarry}
\noeqref{eq:tarry:def}
\begin{align}
    \oJ 
    &\triangleq \limsup_{k \to \infty} \frac{\sum_{n=1}^k J_n}{k}
\label{eq:tarry:def}
 \\ &= \lim_{m \to \infty} \frac{\frac{1}{m} \sum_{i=1}^m \sum_{\ell=1}^{B_i} J_\ell}{\frac{1}{m} \sum_{i=1}^m B_i}
 \label{eq:tarry:blocks}
\\ &= \frac{\E{\sum_{\ell=1}^B J_{\ell}}}{\E{B}}
\label{eq:tarry:ergodicity}
 \\ & = \frac{\E{\sum_{n=2}^{B} (n-1) D_n}}{\E{B}}
\label{eq:tarry:DTOA}
 \\ &= \frac{\E{B^2} - \E{B}}{2q \E{B}} ,
\label{eq:tarry:wald}
\end{align}
\end{subequations}
where \eqref{eq:tarry:blocks} is due to the encoder construction (recall \secref{s:model}), 
\eqref{eq:tarry:ergodicity} follows from ergodicity by recalling that $\{B_i\}_{i=1}^\infty$ and $\{D_n\}_{n=1}^\infty$ are i.i.d.\ and independent of each other,
\eqref{eq:tarry:DTOA} follows from \eqref{eq:tarry:symbol} and exchange of order of summation, 
and \eqref{eq:tarry:wald} follows from Wald's identity and the sum of an arithmetic series. 

\textit{Inter-arrival time:} The inter-arrival time is given by  \eqref{eq:v2f_inter_arrival_time}.

\textit{Service time:} The service time is fixed and is given by \eqref{eq:v2f_service_time}.

\textit{Waiting time:} The waiting time is defined as 
the time that passes 
between a codeword entering the FIFO and the beginning of its service as described in  \secref{s:background}. Bounds on the mean waiting time are provided \secref{s:waiting time}.

The mean delay \eqref{eq:mean-delay} of the system can now be expressed as the sum of the means of three delay elements
\begin{align}
    \bar{\Delta} &= \oJ + \E{S} + \E{W}.
    \label{eq: delay components}
\end{align}
The mean peak AoI \eqref{eq:mean-aoi} can be expressed similarly as 
\begin{align}
    \bar{\Gamma} &= \E{T} + \E{S} + \E{W}.
    \label{eq: AoI components}
\end{align}

%%%%%%%%%%%%%%%%%%%%%%%%%%%%%%%%%%%%%%%%%%%%%%%%%%%%%%%%%%%%%%%%%%%%%%%%%%%%%%%%%%%%%%%%%%%%%%%%%%%%%%%%%%
\iffalse
\section{Periodic Arrival Times}
\label{s:periodic}

In this section we formulate the AoI for periodic arrivals. Recall that the time interval between the arrival of two adjacent symbols is fixed and equal to some known constant $d$. The tarry time for this setup is given by
\begin{align}    
    \E{\bar{J}}=\frac{\E{\sum_{l=1}^K ld}}{\E{K}} = d\frac{\E{K^2}-\E{K}}{2\E{K}}.
\end{align}
The m.g.f of $u$ is given by
\begin{align}
    \phi_U(\theta) & = \E{e^{\theta u}} = \E{e^{\theta(n-\sum_{l=1}^K D_l})} = e^{\theta n}\E{e^{-\theta Kd})}.
\end{align}
The upper bound for the waiting time can now be derived using \eqref{eq:kingman-bound:nu}, by solving the equality 
\begin{align}
    e^{\theta n}\sum_{k=1}^M P(K=k)e^{-\theta kd} = 1
    \label{upper bound periodic}
\end{align}
for real values of $\theta$, with $M\in \mathbb{N}$ being the maximal value that $K$ can attain with a probability larger than $0$.
\fi
%%%%%%%%%%%%%%%%%%%%%%%%%%%%%%%%%%%%%%%%%%%%%%%%%%%%%%%%%%%%%%%%%%%%%%%%%%%%%%%%%%%%%%%%%%%%%%%%%%%%%%%%%%

\subsection{Analysis of the Mean Waiting Time}
\label{s:waiting time}

As was stated in \secref{s:background}, obtaining an analytic expression for the waiting time is difficult, in general.
Instead, we evaluate the two upper bounds of \secref{s:background} for the suggested scheme.

To derive low-moment bound \eqref{eq:kingman-bound:basic}, we first derive an expression for $U$ by substituting 
\eqref{eq:v2f_service_time} and 
\eqref{eq:v2f_inter_arrival_time} in 
\eqref{eq:def:U}:
\begin{align}
    U = \sum_{n=1}^B D_n - \frac{\ell}{\rchannel}.
\end{align}
Thus, the low-moment bound \eqref{eq:kingman-bound:basic} specializes to
\begin{align}
    \E{W} &\leq \frac{\var{B} + (1-q)\E{B}}{q(\E{B}-q\frac{\ell}{\rchannel})} ,
    \label{eq:kingman-bound:basic:v2f:Wald}
\end{align}
with the variance and mean of $T$ given by Wald's identity.

Next, we derive the m.g.f. bound \eqref{eq:kingman-bound:nu}.
To that end, we start by deriving the m.g.f.\ of $U$ \eqref{eq:mgf:U:def}:
\begin{align}
\begin{aligned}
    \phi_U(\theta) 
    &= \E{e^{\theta\left( \frac{\ell}{\rchannel}-\sum_{n=1}^B D_n \right)}} 
 \\ &= e^{\frac{\ell}{\rchannel}\theta} \E{\left(\frac{qe^{-\theta}}{1-(1-q)e^{-\theta}}\right)^B} ,
\end{aligned}
\label{eq:kingman-bound:nu:v2f:substitute:smoothing}
\end{align}
where the second equality follows from the law of total expectation, the fact that $\{D_n\}_{n=1}^\infty$ are i.i.d.\ according to a geometric distribution with parameter $q$ and are independent of $B$, 
and by substituting the m.g.f.\ of $D_n$---$\phi_D(\theta) \triangleq \E{e^{\theta D_n}} = \frac{qe^{\theta}}{1-(1-q)e^{\theta}}$.
Because $\phi_U(\theta)$ is a continuous function, the bound can be found by solving the equation 
\begin{align}
& \phi_U(\theta) = 1,
\end{align}
which, by \eqref{eq:kingman-bound:nu:v2f:substitute:smoothing}, can be rewritten as
\begin{align}
&\sum_{b=1}^{b_{\max}}\left(\frac{q \e^{-\theta}}{1-(1-q) \e^{-\theta}}\right)^b P(B=b) = \e^{-\frac{\ell}{\rchannel}\theta},
\end{align}
or, equivalently, as
\begin{align} 
    & \left({1-(1-q) \e^{-\theta}}\right)^{b_{\max}} \e^{-\theta \frac{\ell}{\rchannel}}
\\* &\qquad -\sum_{b=1}^{b_{\max}}P(B=b) \left( {1-(1-q) \e^{-\theta}} \right)^{b_{\max} - b } q \e^{-\theta b}  = 0,
\end{align}
with $b_{\max}$ being the maximal value that $B$ can attain with a non-zero probability. 
The solutions for this equation can be found using an appropriate root finding algorithm, \eg, the Newton--Raphson algorithm \cite{gil2007numerical}.\footnote{If $\frac{\ell}{\rchannel}$ is integer, the equation can be solved by substituting $z=e^{-\theta}$ and solving the resulting polynomial equation.} Once the solutions are found, the bound is derived by taking the minimum of the real solutions, and substituting it in \eqref{eq:kingman-bound:nu}.\footnote{The bound is guaranteed to exist since $\phi_U(0)=1$. In the case that this is the only real solution, the bound trivializes to $\infty$.}

%%%%%%%%%%%%%%%%%%%%%%%%%%%%%%%%%%%%%%%%%%%%%%%%%%%%%%%%%%%%%%%%%%%%%%%%%%%%%%%%%%%%%%%%%%%%%%%%%%%%%%%%%%

\section{Simulation Study}
\label{s:numeric}

In this section, we simulate variable-to-fixed Tunstall codes and compare their performance to those of the optimized fixed-to-variable codes of \cite{zhong2016timeliness}. 

 To that end, we generated several simulations of $2\times10^6$ i.i.d.\ Bernoulli distributed source samples of $\{X_n | n= 1, \ldots, 2\times10^6\}$, each simulation with a different parameter $p \triangleq \Pr(X_i=1) \in (0,1/2]$, corresponding to a different source entropy. 
 Note that an increase in the entropy corresponds to an increase in traffic since a higher (average) rate is required to describe a source sample. The probability of arrival in all simulations was set to $q=1/2$.

To encode the source, we used an off-the-shelf Tunstall code for each parameter $p$, which minimizes the code rate $\rcode$ \eqref{eq:v2f:code-rate} for a given $\ell$; note that the queue is stabilizable with some variable-to-fixed code with a given $\ell$ iff it is stabilizable with a Tunstall code with this $\ell$ for the given distribution. The output length was set to $\ell = 4$, which is the minimal length for the chosen region of $p$ that satisfies \eqref{eq:system_stability}, and therefore minimizes the service time.

In order to test the performance of the waiting time bounds, we compare between the empiric results and the analytic expression of the delay given in \eqref{eq: delay components}, with the waiting time bounded as described in \secref{s:waiting time}.

We compare our results to a fixed-to-variable scheme adapted from a work by Zhong \etal~\cite{zhong2016timeliness}. To that end, we derive a bound for the waiting time using \eqref{eq:kingman-bound:basic}
\begin{align}
    \E{W} \leq \frac{\var{L}/\rchannel^2+(1-q)b/q^2}{b/q-\E{L}/\rchannel}.
\end{align}
 The mean delay is then bounded by 
\begin{align}
    \bar{\Delta} \leq \frac{\var{L}/\rchannel^2+(1-q)b/q^2}{b/q-\E{L}/\rchannel} + \frac{\E{L}}{\rchannel} + \frac{b-1}{2q}.
\end{align}
The mean peak AoI is bounded similarly by 
\begin{align}
    \bar{\Gamma} \leq \frac{\var{L}/\rchannel^2+(1-q)b/q^2}{b/q-\E{L}/\rchannel} + \frac{\E{L}}{\rchannel} + \frac{b}{q}.
\end{align}
Note that a fixed-to-variable code that minimizes each of the two bounds will necessarily minimize the other as they differ only by a constant. Such a fixed-to-variable code was found using Larmore's convex hull algorithm \cite{larmore1989minimum}. 
The block size was set to $b=4$, which is the minimal block size that maintains stability for the chosen region of $p$ (as per \cite{zhong2016timeliness}).

The mean delay $\bar{\Delta}$ \eqref{eq:mean-delay} and mean peak AoI $\bar{\Gamma}$ \eqref{eq:mean-aoi} versus $H(X)$ of each scheme are depicted in \figref{fig: random arrivals delay and aoi}, along with the analytical upper bounds of \secref{s:AoI-V2F}.

As we can see, the variable-to-fixed Tunstall code outperforms the fixed-to-variable coding scheme in the high-traffic region, while the fixed-to-variable scheme offers better results in the low-traffic region. 
Furthermore, we observe that the low-moment bound, while easier to obtain than the m.g.f.\ bound, is less tight, especially in the high-traffic region. 
\begin{figure}[t]
    \centering
    \includegraphics[width=\columnwidth]{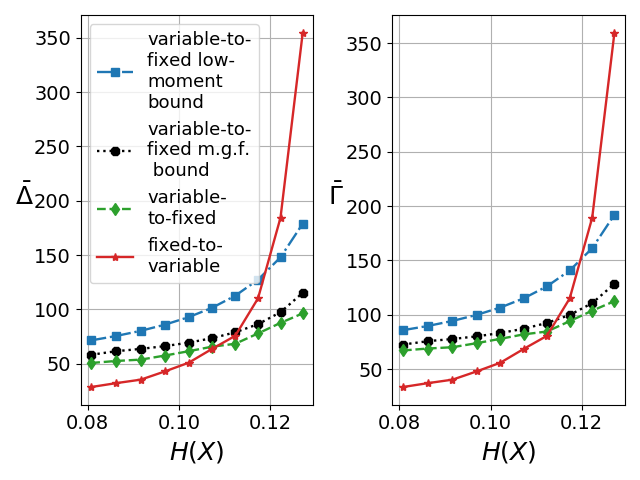}
    \caption{Mean delay and mean peak AoI as a function of source entropy for random arrivals.
    $\rchannel = 1/6.5$. fixed-to-variable block size $b=4$. variable-to-fixed output length $\ell = 4$.}
    \label{fig: random arrivals delay and aoi}
\end{figure}

%%%%%%%%%%%%%%%%%%%%%%%%%%%%%%%%%%%%%%%%%%%%%%%%%%%%%%%%%%%%%%%%%%%%%%%%%%%%%%%%%%%%%%%%%%%%%%%%%%%%%%%%%%

\section{Discussion and Future Work}
\label{s:summary}

We put forward variable-to-fixed coding for real-time source coding as an alternative to  the hitherto used fixed-to-variable coding.
We analyze our proposed method and demonstrate that it outperforms its fixed-to-variable counterparts when traffic is high both in terms of delay and AoI.
We suggest the following topics for future work which are currently under investigation:
\begin{itemize}\addtolength\itemsep{.25\baselineskip}
\item 
    We have used off-the-shelf Tunstall codes and demonstrated that even such codes outperform optimized fixed-to-variable codes when traffic is high (large $H(X)$). However, although these codes are optimal in terms of stabilizability (\cf\ Huffman codes in fixed-to-variable coding) they are not necessarily optimal in terms of minimum mean delay (\cf\ optimized codes of Zhong \etal~\cite{zhong2017backlog}). Designing minimum mean delay codes is an interesting research avenue.
\item 
    Extension of the problem setup to allow several symbols arriving at the same time instant.
\item
    As suggested by \secref{s:numeric}, variable-to-fixed codes outperform fixed-to-variable codes for some parameters, and vice versa.
    Designing variable-to-variable coding schemes that 
    outperform both of these classes of codes holds promise. This approach was explored in \cite{zhong2017backlog} where an adaptive choice of the block size was investigated for the simpler setting of periodic (deterministic) arrivals. 
\item 
    Throughout this work, an implicit indicator signal was assumed available at the decoder that states whether the FIFO queue is empty or not.
    Studying this problem in the absence of such an indicator signal, in which case the event of not enough (or none at all) source symbols to encode needs to be taken into account (and encoded). 
    This has been done for fixed-to-variable coding in \cite{ZhongYatesSoljanin:Huffman_for_AoI:ITW2018}.
\item 
    Formulation of the m.g.f.\ waiting-time bound for fixed-to-variable codes to use in lieu of the low-moment bound for the design of optimal fixed-to-variable codes.
\end{itemize}

%%%%%%%%%%%%%%%%%%%%%%%%%%%%%%%%%%%%%%%%%%%%%%%%%%%%%%%%%%%%%%%%%%%%%%%%%%%%%%%%%%%%%%%%%%%%%%%%%%%%%%%%%%

\section{Acknowledgments}
\label{s:acknowledgements}
We thank Emina Soljanin and Jing Zhong for supplying the code generating the figures in \cite{ZhongYatesSoljanin:Huffman_for_AoI:ITW2018} and for helpful discussions.
%%%%%%%%%%%%%%%%%%%%%%%%%%%%%%%%%%%%%%%%%%%%%%%%%%%%%%%%%%%%%%%%%%%%%%%%%%%%%%%%%%%%%%%%%%%%%%%%%%%%%%%%%%

% \ver{}{\vspace{-.47\baselineskip}}
\bibliographystyle{IEEEtran}
\bibliography{toly, uabend}

%%%%%%%%%%%%%%%%%%%%%%%%%%%%%%%%%%%%%%%%%%%%%%%%%%%%%%%%%%%%%%%%%%%%%%%%%%%%%%%%%%%%%%%%%%%%%%%%%%%%%%%%%%
%%%%%%%%%%%%%%%%%%%%%%%%%%%%%%%%%%%%%%%%%%%%%%%%%%%%%%%%%%%%%%%%%%%%%%%%%%%%%%%%%%%%%%%%%%%%%%%%%%%%%%%%%%
%%%%%%%%%%%%%%%%%%%%%%%%%%%%%%%%%%%%%%%%%%%%%%%%%%%%%%%%%%%%%%%%%%%%%%%%%%%%%%%%%%%%%%%%%%%%%%%%%%%%%%%%%%
\end{document}